\documentclass{article}
\usepackage[margin=1in]{geometry}

\usepackage{graphicx}
\usepackage{amsthm,amsmath,amssymb}
\usepackage{cite}
\usepackage{tikz}
\usepackage{subcaption}
\usetikzlibrary{patterns,intersections}

\usepackage[hidelinks]{hyperref}
\usepackage{cleveref}

\newtheorem{lemma}{Lemma}
\newtheorem{remark}{Remark}

\newtheorem{definition}{Definition}
\newtheorem{thm}{Theorem}

\crefname{thm}{Theorem}{Theorems}
\crefrangelabelformat{thm}{#3#1#4--#5#2#6}

\newcommand{\ci}{\mathcal{C}^I}
\newcommand{\F}{\mathbb{F}_q}
\newcommand{\cf}{\mathcal{C}^F}

\newcommand{\gf}{\mathbf{G}^F}
\newcommand{\gi}{\mathbf{G}^I}
\newcommand{\rf}{r^F}
\newcommand{\ri}{r^I}
\newcommand{\kf}{k^F}
\newcommand{\ki}{k^I}
\newcommand{\nf}{n^F}
\newcommand{\nii}{n^I}

\newcommand{\rank}{\mathrm{rank}}

\newcommand{\Gi}{
\begin{bmatrix}
        \mathbf{I}_\ell&&&\mathbf{B}_{1,1}&\cdots&\mathbf{B}_{1,r^{I}}\\
        &\ddots&&\vdots&\ddots&\vdots\\
        &&\mathbf{I}_\ell&\mathbf{B}_{\ki,1}&\cdots&\mathbf{B}_{\ki,r^{I}}\\
    \end{bmatrix}}

\newcommand{\Gf}{
\begin{bmatrix}
        \mathbf{I}_\ell&&&\mathbf{C}_{1,1}&\cdots&\mathbf{C}_{1,\rf}\\
        &\ddots&&\vdots&\ddots&\vdots\\
        &&\mathbf{I}_\ell&\mathbf{C}_{\kf,1}&\cdots&\mathbf{C}_{\kf,\rf}\\
    \end{bmatrix}}

\newcommand{\B}{
\begin{bmatrix}
        \mathbf{B}_{1,1}&\cdots&\mathbf{B}_{1,r^{I}}\\
        \vdots&\ddots&\vdots\\
        \mathbf{B}_{\ki,1}&\cdots&\mathbf{B}_{\ki,r^{I}}\\
    \end{bmatrix}}

\newcommand{\C}{
\begin{bmatrix}
        \mathbf{C}_{1,1}&\cdots&\mathbf{C}_{1,\rf}\\
        \vdots&\ddots&\vdots\\
        \mathbf{C}_{\kf,1}&\cdots&\mathbf{C}_{\kf,\rf}\\
    \end{bmatrix}}

\begin{document}

\title{Lower Bounds on Conversion Bandwidth for MDS Convertible Codes in Split Regime}

\author{
Lewen Wang and Sihuang Hu
\thanks{The authors are with the State Key Laboratory of Cryptography and Digital Economy Security, the Key Laboratory of Cryptologic Technology and Information Security, Ministry of Education, and the School of Cyber Science and Technology, Shandong University, Qingdao, Shandong 266237, China (e-mail: lewenwang3@gmail.com; husihuang@sdu.edu.cn). Corresponding author: Sihuang Hu.}
}

\date{}

\maketitle
\begin{abstract}
We propose several new lower bounds on the bandwidth cost of MDS convertible codes using a linear-algebraic framework. The derived bounds improve previous results in certain parameter regimes and match the bandwidth cost of the construction proposed by Maturana and Rashmi (2022 IEEE International Symposium on Information Theory) for $r^F\le r^I\le k^F$, implying that our bounds are tight in this case.
\end{abstract}


\section{Introduction}
Erasure codes are widely used in distributed storage systems as they provide fault tolerance with smaller storage overhead compared to replication~\cite{weatherspoon02}. 
In a typical system, a file  is  divided  into $k$ data symbols  and then encoded  into $n$   symbols using an $[n,k]$ erasure code. This encoding fixes the fault-tolerance level of the system.   
However, large-scale storage systems, such as those operated by Google and other cloud providers, contain storage nodes whose failure probabilities vary over time. 
Prior work by Kadekodi \emph{et al.}~\cite{kadekodi2019cluster} has demonstrated that dynamically adjusting the code parameters to match the observed changes in node failure rates can lead to substantial savings in storage overhead. 
For instance, tailoring $n$ and $k$ to the current failure environment may reduce storage requirements by over 11–44\%.

Therefore, it is important to support efficient conversion of commonly used codes---particularly MDS codes, which provide maximal reliability for a given storage overhead---so that the code can adapt to changing system reliability requirements.  
However, naively adjusting the code rate requires fully re-encoding all stored data, which is both computationally and I/O intensive.
To address this issue, Maturana and Rashmi~\cite{maturana2020convertible,maturana2022convertible} introduced the framework of \emph{convertible codes}, which allows an initial code with  given parameters to be converted efficiently into a final code with  different parameters, avoiding full re-encoding. The conversion process transforms codewords in the initial code into codewords in the final code while preserving the original information.
An \emph{MDS convertible code} refers to a convertible code in which both the initial and final codes are MDS codes.

There are two fundamental regimes of code conversion: the \emph{merge regime}, in which multiple initial codewords are merged into a single final codeword, and the \emph{split regime}, in which one initial codeword is divided into multiple final codewords.
These two regimes capture the essential trade-offs in the general convertible code framework.

The efficiency of a conversion process is typically measured by two metrics.  
The first is the \emph{access cost}, defined as the total number of coded symbols accessed during conversion.  
Maturana and Rashmi established  tight lower bounds on the access cost of MDS convertible codes  in  both   merge and split regimes and proposed access-optimal constructions that achieve these bounds~\cite{maturana2020convertible,maturana2022convertible,Maturana2020isit}.  
Subsequent works~\cite{maturana2020convertible,Chopra2024OnLF,Kong2024} have focused on reducing the field size required for such constructions, while others have extended the analysis to different regimes and code classes~\cite{Kong2024,Maturana2023isit,Ge2024MDSGC,Shi2025BoundsAO,Ge2025LocallyRC,RKSVK25,Gruica2026ConvertibleCF,Zhang25}.

The second performance metric is the \emph{bandwidth cost}, which measures the total amount of data transferred between nodes during conversion.  
In~\cite{Maturana2023}, Maturana and Rashmi derived a tight lower bound on the bandwidth cost of MDS convertible codes in the merge regime and proposed a bandwidth-optimal construction that attains this bound.  
For the split regime, Maturana and Rashmi~\cite{Maturana2022BandwidthCO} proposed a lower bound on bandwidth cost based on an information-flow model. They also introduced a conjecture and left the problem of determining the minimum bandwidth cost as an open question.
Very recently, Singhvi et al.~\cite{singhvi2025tight} employed an information-theoretic framework to derive lower bounds on the bandwidth cost of MDS convertible codes in the split regime. 
We note that their methodology is distinct from the one pursued in this work.


\subsection{Our Contribution}

In this paper, we establish lower bounds on the bandwidth cost of MDS convertible codes with linear conversion procedures in the split regime.  
The main contributions are summarized as follows:
\begin{itemize}
    \item \textbf{Linear-algebraic reformulation.}  
    We introduce a vector-space perspective on code conversion by identifying an inclusion relation between specific column spaces of the generator matrices of the initial code and the final code.  
    This reformulation converts the problem of minimizing bandwidth cost into a linear-algebraic optimization problem.

    \item \textbf{Closed-form lower bounds.}  
    Building on this inclusion relation, we derive explicit closed-form lower bounds on the total read bandwidth by solving a family of linear programs.  
    The resulting expressions are formally presented in \cref{bound0,bound1,bound_3}.

    \item \textbf{Comparison with prior work.}  
    Our framework removes the assumption of uniform data download across unchanged and retired symbols in~\cite{Maturana2022BandwidthCO}. Moreover, our lower bounds are strictly tighter than those in Theorem~4 of~\cite{Maturana2022BandwidthCO} for most parameter regimes. In addition, the bound in~\cref{bound1} coincides with the bandwidth cost achieved by their construction for $r^F \le r^I \le k^F$, which proves that our bound is \emph{tight} in this case. A detailed comparison between our results and those in Theorem~4 of~\cite{Maturana2022BandwidthCO} is presented in Section~\ref{conclusion}.
\end{itemize}

\subsection{Preliminaries}
We first introduce some basic definitions and notations.
Let $\mathbb{F}_q$ be a finite field. 
An $[n,k,\ell]$ \emph{array code} $\mathcal{C}$ is a subspace of $\mathbb{F}_q^{n\ell}$ of dimension $k\ell$. 
Each codeword $\mathbf{c}$ $\in \mathcal{C}$ is represented as 
\[
\mathbf{c} = (\mathbf{c}_1, \dots, \mathbf{c}_n)^{T}, \quad 
\mathbf{c}_i = (c_{i,1}, \dots, c_{i,\ell}) \in \mathbb{F}_q^\ell, \quad i \in [n].
\]
In the following context, we refer to $\mathbf{c}_i$ as a (codeword) \emph{symbol}, and each scalar $c_{i,j}$ as a \emph{subsymbol}. We call an $[n,k,\ell]$ array code $\mathcal{C}$ an \emph{MDS array code} if any $k$ out of $n$ symbols suffice to recover the whole codeword.

A generator matrix $\mathbf{G}$ of  $\mathcal{C}$  is a $k\ell \times  n\ell$ matrix over $\F$ whose rows form a basis for the code. 
The generator matrix is said to be systematic if it has the block form 
$$\mathbf{G}=[\mathbf{I}_{k\ell}\,|\,\mathbf{A}],$$
where $\mathbf{I}_{k\ell}$ is the $k\ell\times k\ell$ identity  matrix and  $\mathbf{A}$ is a $k\ell\times r\ell$ matrix. 
 For a message vector $\mathbf{m}\in\F^{k\ell}$, the encoded codeword under $\mathbf{G}$ is denoted by $\mathcal{C}(\mathbf{m})=\mathbf{m}\mathbf{G}$. 

We recall the definition of MDS convertible codes in the split regime.
\begin{definition}[MDS Convertible Codes~\cite{maturana2022convertible}]\label{def_convertible_code}
    Let $\lambda$ be an integer with $\lambda\ge 2$. An $[n^{I},k^{I}=\lambda\kf;n^{F},k^{F};\ell]$  MDS convertible code over a finite field $\F$ can be defined as 
    \begin{itemize}
        \item A pair of   codes $(\ci,\cf)$ where  $\ci$ is an initial  MDS array code with parameter $[\nii,\ki,\ell]$ and  $\cf$ is a final  MDS array code with parameter $[\nf,\kf,\ell]$.
        \item A conversion procedure $\mathcal{T}$ with input $\{\ci(\mathbf{m}):\mathbf{m}=(\mathbf{m}_1,\cdots,\mathbf{m}_\lambda)\}$  and output $\{\cf(\mathbf{m}_i):i\in [\lambda]=\{1,\cdots,\lambda\}\}$ for all  $\mathbf{m}_i\in \F^{\kf\ell}$.
    \end{itemize}
\end{definition}

During conversion, each symbol of the initial and final codewords  belongs to one of three categories:
(1) Unchanged symbols: appear in both the initial and final codewords. 
(2) Retired symbols: appear only in the initial codeword.
(3) New symbols: appear only in the final codewords.
For $i\in [\lambda]$, let  $N_i$ be the set of indices of new symbols  in the $i$-th final codeword.

For each initial symbol index $j\in[n^I]$ let $D_j\subseteq[\ell]$ denote the set of subsymbol indices that are read from symbol $\mathbf{c}_j$ during conversion, and write $\bar{D}_j=[\ell]\setminus D_j$ for the unread subsymbol indices.  
Define $\beta_i=|D_i|$ as the number of subsymbols read from symbol $\mathbf{c}_i$.
Then,
\begin{itemize}
    \item The  read bandwidth cost  is  $R=\sum\limits_{i=1}^{\nii}\beta_i$. 
    \item The write bandwidth cost is   $W=(\sum\limits_{i=1}^\lambda|N_i|)\ell$.
    \item The total bandwidth cost is $R+W$.
\end{itemize}

Intuitively, having more unchanged symbols leads to lower write bandwidth cost.
A convertible code is  stable if it maximizes the number of unchanged symbols for its parameter set.
\begin{definition}[Stable Convertible Code\cite{maturana2022convertible}]\label{def_stable_convertible}
    An $[n^{I},k^{I};n^{F},k^{F};\ell]$  MDS convertible code is said to be \emph{stable} if it uses the maximum number of unchanged symbols over all  MDS convertible codes with the same parameters.
\end{definition}
As for an MDS convertible code in the split regime, we have the following result.
\begin{lemma}\label{lem_stable}
    Let $(\mathcal{C}^I,\mathcal{C}^F)$ be an $[n^{I},k^{I};n^{F},k^{F};\ell]$ MDS convertible code in the split regime with $k^I=\lambda k^F$. Then the number of unchanged symbols is at most $k^I$, and this bound is achievable.
\end{lemma}
\begin{proof}

    By the MDS property of $\mathcal{C}^F$, any subset of $k^F + 1$ symbols is linearly dependent. 
    Hence, each final codeword can contain at most $k^F$ unchanged symbols from the initial codeword. Otherwise, these $k^F + 1\le\ki$ symbols are linearly dependent in  the initial codeword, which contradicts the MDS property of $\mathcal{C}^I$.
    Since the initial codeword is split into $\lambda$ final codewords, the total number of unchanged symbols is at most $\lambda k^F = k^I$. 
    The number of unchanged symbols can be achieved by straightforward re-encoding.
\end{proof}

For a matrix $\mathbf{M}$, we denote by $\langle \mathbf{M}\rangle$ the column space of $\mathbf{M}$.  
Let $S_1\subseteq [m], S_2\subseteq [n]$. If  the matrix $\mathbf{M}$ has size $m\times n$, we use $\mathbf{M}[S_1;S_2]$ to denote the submatrix of $\mathbf{M}$ formed by selecting rows indexed by $S_1$ and columns indexed by $S_2$.
If $\mathbf{M}$ is a  block matrix with $m\ell\times n\ell$ entries and each block of size $\ell\times\ell$,
we write $\mathbf{M}[S_1;S_2]$ to be the block submatrix consisting of block rows indexed by  $S_1$ and block columns indexed by $S_2$.
For brevity, we also write $\mathbf{M}[S_1; :]$ (resp. $\mathbf{M}[:;S_2]$) to denote the submatrix obtained by selecting only the block rows indexed by $S_1$ (resp. only the block columns indexed by $S_2$).

\subsection{Organization of This Paper}
The remainder of this paper is organized as follows.  In Section~\ref{inclusion_relation}, we establish an inclusion relation between the column spaces of the generator matrices of the initial and final codes, which forms the algebraic foundation for our lower-bound analysis.
In Section~\ref{bounds}, we derive the main results — lower bounds on the  bandwidth cost of MDS convertible codes in the split regime. Finally, in Section~\ref{conclusion}, we conclude the paper with a comparison between our bound and the existing results of Maturana and Rashmi~\cite{Maturana2022BandwidthCO}. An explicit example achieving our bound is presented in  Appendix~\ref{ppendix_A}.

\section{An Inclusion Relation Between Generator Matrices}\label{inclusion_relation}
In this section, we establish an inclusion relation between column spaces of the systematic generator matrices of the initial and final codes during conversion.
This relation serves as a crucial algebraic foundation for deriving the lower bounds on bandwidth cost in Section~\ref{bounds}.

As in~\cite{Maturana2022BandwidthCO}, we focus on stable convertible codes. By Lemma~\ref{lem_stable} and Definition~\ref{def_stable_convertible}, we have $|N_i|=\rf$ for each $i\in[\lambda]$. Thus, the write bandwidth cost is fixed as $W=(\sum_{i=1}^\lambda|N_i|)\ell=\lambda\rf\ell$. Minimizing the total bandwidth cost therefore reduces to minimizing the read bandwidth cost.   We next specify the structure of the generator matrices of the initial and final codes.

Assume the first $\ki$ symbols of the initial code and the first $\kf$ symbols of the final code are stable. Then, the generator  matrices of $\ci,\cf$ can be written in the following systematic form. 
     $$\gi=\Gi, \mathbf{B}=\B.$$
     $$\gf=\Gf, \mathbf{C}=\C.$$\\
Here each block $\mathbf{B}_{i,j}(i\in[\ki],j\in[\ri])$ and $\mathbf{C}_{i,j}(i\in[\kf],j\in[\rf])$ is an $\ell\times\ell$ matrix over $\F$.
 Both $\mathcal{C}^I$ and $\mathcal{C}^F$ are MDS array codes if and only if every block square submatrix of $\mathbf{B}$ and $\mathbf{C}$ is nonsingular. This can be viewed as a block-matrix extension of the superregular matrices defined in Section II.B of \cite{maturana2022convertible}.
 
The conversion and its associated bandwidth cost can be characterized by the following lemma.
\begin{lemma}\label{character_lemma}
    Let $(\mathcal{C}^I, \mathcal{C}^F)$ be a stable $[n^{I},k^{I}=\lambda k^{F};n^{F},k^{F};\ell]$ convertible code  with  generator matrices $\gi$ and $\gf$ as defined above. Let $\mathcal{T}$ denote the linear conversion procedure that minimizes  the read cost. \\
Write  \begin{equation*}
        \tilde{\mathbf{C}}=
        \begin{bmatrix}
            \mathbf{C}^{(1)} & & \\
            & \ddots & \\
            & & \mathbf{C}^{(\lambda)}
        \end{bmatrix}
        , 
    \end{equation*}
     where
     \begin{equation*}
     \mathbf{C}^{(i)} = 
    \begin{bmatrix}
        \mathbf{C}_{1,1}[\overline{D_{(i-1)k^F+1}}; :] & \cdots & \mathbf{C}_{1, r^F}[\overline{D_{(i-1)k^F+1}}; :] \\
        \vdots & \ddots & \vdots \\
        \mathbf{C}_{k^F,1}[\overline{D_{ik^F}}; :] & \cdots & \mathbf{C}_{k^F, r^F}[\overline{D_{ik^F}}; :]
    \end{bmatrix},
     \end{equation*}
    for $i \in [\lambda]$, and
    \begin{equation*}
        \tilde{\mathbf{B}}=
        \begin{bmatrix}
            \mathbf{B}_{1,1}[\overline{D_{1}}; D_{\ki+1}] & \cdots & \mathbf{B}_{1,\ri}[\overline{D_{1}}; D_{n^I}] \\
            \vdots & \ddots & \vdots \\
            \mathbf{B}_{\ki,1}[\overline{D_{\ki}}; D_{\ki+1}] & \cdots & \mathbf{B}_{\ki,\ri}[\overline{D_{\ki}}; D_{n^I}]
       \end{bmatrix}.
    \end{equation*}
 Then it holds that:
    \begin{equation}\label{character_eq}
        \left\langle\tilde{\mathbf{C}} \right\rangle
        \subseteq
        \left\langle\tilde{\mathbf{B}} \right\rangle.
    \end{equation}
    Moreover, the matrix $\tilde{\mathbf{B}}$ has full column rank.
\end{lemma}

\begin{proof}
    Since $(\ci,\cf)$ forms a stable convertible code  with generator matrices $\gi$ and $\gf$, the conversion procedure $\mathcal{T}$ guarantees that,  for each $i\in[\lambda]$ and any message  $\mathbf{m}_i\in\F^{\kf\ell}$, the new symbols in the $i$-th final codeword $\mathcal{C}^F(\mathbf{m}_i)$ are computable solely from  subsymbols read from the initial codeword $\mathcal{C}^I(\mathbf{m}_1,\cdots,\mathbf{m}_\lambda)$. By Definition~\ref{def_convertible_code}, there exists a matrix \( \mathbf{T} \) such that
    \begin{equation}\label{eq1_character_lem}
    (\mathbf{m}_1,\cdots,\mathbf{m}_\lambda)
    \begin{bmatrix}
        \mathbf{C}&&\\&\ddots&\\&&\mathbf{C}
    \end{bmatrix}
    =(\mathbf{m}_1,\cdots,\mathbf{m}_\lambda)\tilde{\mathbf{G}}^I \mathbf{T},
    \end{equation}
    where $\tilde{\mathbf{G}}^I$ is the submatrix of $\mathbf{G}^I$ formed by the columns corresponding to the read subsymbols, i.e.,
    \[
\tilde{\mathbf{G}}^I =
\left[
\mathbf{D}_\mathcal{T}\ \middle|\ \mathbf{B}_\mathcal{T}
\right],
\]
where
\[
\mathbf{D}_\mathcal{T}
=
\operatorname{diag}\big(
\mathbf{I}_\ell[:,D_1],
\ldots,
\mathbf{I}_\ell[:,D_{k^I}]
\big),
\]
and
\[
\mathbf{B}_\mathcal{T}
=
\begin{bmatrix}
\mathbf{B}_{1,1}[:,D_{k^I+1}] & \cdots & \mathbf{B}_{1,r^I}[:,D_{n^I}] \\
\vdots & \ddots & \vdots \\
\mathbf{B}_{k^I,1}[:,D_{k^I+1}] & \cdots & \mathbf{B}_{k^I,r^I}[:,D_{n^I}]
\end{bmatrix}.
\]
Since \eqref{eq1_character_lem} holds for all message vectors $\mathbf{m}=(\mathbf{m}_1,\cdots,\mathbf{m}_\lambda)$, one can choose
    $\mathbf{m}$ ranging over all standard basis vectors of $\mathbb{F}_q^{k^I \ell}$. In that case, each standard basis vector selects the corresponding row of both the left-hand block-diagonal matrix and the right-hand matrix $\tilde{\mathbf{G}}^I \mathbf{T}$. Then, we have, 
    \begin{align}\label{character_eq3}
\left[
\begin{array}{ccc}
\mathbf{C} & & \\
& \ddots & \\
& & \mathbf{C}
\end{array}
\right] = \tilde{\mathbf{G}}^I \mathbf{T}.
\end{align}
This implies that
    \begin{equation}\label{character_eq2}
         \left\langle
         \begin{bmatrix}
            \mathbf{C}&&\\&\ddots&\\&&\mathbf{C}
         \end{bmatrix}\right\rangle\subseteq\left\langle
         \tilde{\mathbf{G}}^I\right\rangle.
    \end{equation}
    By eliminating all rows corresponding to the unique nonzero entries of the identity sub-blocks in $\tilde{\gi}$, we  obtain  inclusion  \eqref{character_eq}.

  To prove that $\tilde{\mathbf{B}}$ has full column rank, assume for contradiction that it does not.
  Then some nontrivial linear combination of its columns equals zero, implying a nontrivial dependence among the columns of $\tilde{\mathbf{G}}^I$.  
  This in turn means that certain read subsymbols are linearly dependent, and hence at least one of them is redundant for reconstruction. 
  However, because we assume that the conversion procedure $\mathcal{T}$ minimizes the read cost, it does not carry out any unnecessary reads.
  This contradiction shows that $\tilde{\mathbf{B}}$ has full column rank.
\end{proof}

\begin{remark}\label{compute_bandwidth}
Conversely, if the inclusion relation~\eqref{character_eq} holds, 
then \eqref{character_eq2} follows directly. Then, there exists a matrix $\mathbf{T}$ such that \eqref{character_eq3} holds. This matrix $\mathbf{T}$ induces a conversion procedure $\mathcal{T}$ with read bandwidth cost 
$$
\rank(\tilde{\mathbf{B}}) + k^I\ell - \mathrm{row}(\tilde{\mathbf{B}}),
$$
where $\mathrm{row}(\tilde{\mathbf{B}})$ denotes the number of rows of $\tilde{\mathbf{B}}$. This is because the reads from the first $k^I$ systematic symbols contribute $\sum_{j=1}^{k^I}|D_j|$ subsymbols, and 
\[
\mathrm{row}(\tilde{\mathbf{B}})=\sum_{j=1}^{k^I}|\overline{D_j}|
=\sum_{j=1}^{k^I}(\ell-|D_j|)=k^I\ell-\sum_{j=1}^{k^I}|D_j|,
\]
so $\sum_{j=1}^{k^I}|D_j|=k^I\ell-\mathrm{row}(\tilde{\mathbf{B}})$.
The remaining downloads come from the last $r^I$ symbols and correspond to the columns of $\tilde{\mathbf{B}}$. If these downloaded columns are linearly dependent, then some downloads are redundant and can be removed without affecting the induced linear relation. Hence it suffices to download only $\rank(\tilde{\mathbf{B}})$ subsymbols (equivalently, the corresponding columns) that span $\langle\tilde{\mathbf{B}}\rangle$. Therefore the read bandwidth cost is $\rank(\tilde{\mathbf{B}})+k^I\ell-\mathrm{row}(\tilde{\mathbf{B}})$.

\end{remark}

\section{ Lower Bounds on Bandwidth Cost}\label{bounds}

In this section, under the assumption that the conversion procedure is linear, we provide several lower bounds on the read bandwidth cost of stable MDS convertible codes in the split regime.
\begin{thm}\label{bound0}
    For every stable  $[\nii,\ki=\lambda\kf;\nf,\kf;\ell]$MDS convertible code with $\kf\le\rf$, the read bandwidth cost satisfies
    \begin{equation*}
        R \geq \ki\ell.
    \end{equation*}
\end{thm}

\begin{proof}
    Let $(\ci,\cf)$ be a stable convertible code with $\kf\le\rf$, and a conversion procedure $\mathcal{T}$ achieving the minimum bandwidth cost. 
    By \cref{character_lemma}, the inclusion relation~\eqref{character_eq} implies that $\rank(\tilde{\mathbf{C}})\le \rank(\tilde{\mathbf{B}})$.
    Since $\tilde{\mathbf{B}}$ has full column rank, we obtain
    \begin{equation}\label{bound1_total_rank}
        \sum\limits_{i=1}^{\lambda}\rank(\mathbf{C}^{(i)})\le \sum\limits_{j=\ki+1}^{\nii} \beta_j.
    \end{equation}
    Since $\kf\le\rf$, the block matrix $\mathbf{C}$ has full row rank and $\rank(\mathbf{C})=\kf\ell$. Hence, for each $i$,
    $$\rank(\mathbf{C}^{(i)})\ge \kf\ell-\sum_{j=(i-1)\kf+1}^{i\kf}\beta_j.$$
    Summing over $i\in[\lambda]$, we have
     \begin{equation*}
\begin{aligned}
\sum_{i=1}^{\lambda}\rank(\mathbf{C}^{(i)})
&\ge
\sum_{i=1}^{\lambda}
\left(
\kf\ell
-
\sum_{j=(i-1)\kf+1}^{i\kf}\beta_j
\right)  \\
&=
\ki\ell-\sum_{j=1}^{\ki}\beta_j .
\end{aligned}
\end{equation*}
   Combining this with~\eqref{bound1_total_rank} yields that
    $$\sum\limits_{j=\ki+1}^{\nii} \beta_j\ge\lambda\kf\ell-\sum_{j=1}^{\ki}\beta_j.$$
    It follows that ${R} \geq \ki\ell.$
\end{proof}
\begin{remark}
    The lower bound $\lambda\kf\ell$ can  be achieved by full re-encoding and is therefore tight. This coincides with the lower bound in Theorem~4 of~\cite{Maturana2022BandwidthCO} for the case $k^F \le r^F$, but is derived here via a distinct algebraic argument.
\end{remark}

For the case where $\kf \ge \rf$, some additional structural constraints on the generator matrices of the initial and final codes arise, leading to another lower bound on the read bandwidth cost, as stated below.

\begin{thm}\label{bound1}
    For every stable MDS $[\nii,\ki=\lambda\kf;\nf,\kf;\ell]$ convertible code satisfying $\rf<\kf$ and $\ri\le\kf$, the  read bandwidth cost satisfies
    \begin{equation*}
        R \ge\begin{cases}
            \lambda \kf \ell-\frac{(\kf-\rf)\ri}{\rf}\ell & \text{if } \ri\le\rf\le\kf,\\[2mm]
            \lambda\rf\frac{(\lambda-1)\kf+\ri}{(\lambda-1)\rf+\ri}\ell, &\text{if } \rf\le\ri\le\kf.
        \end{cases} 
    \end{equation*}
\end{thm}
\begin{proof}
    Let $(\ci,\cf)$ be a stable MDS convertible code with $\rf<\kf$, $\ri\le\kf$ and a conversion procedure $\mathcal{T}$ that minimizes the total bandwidth cost. 
    By \cref{character_lemma}, the inclusion relation~\eqref{character_eq}  holds, and so does~\eqref{bound1_total_rank}. 
    
    Consider any subset $U_i=\{u_{i,1},\cdots,u_{i,{\rf}}\}\subseteq[\kf]$ of size $\rf$. Since $\cf$ is an MDS code, the square submatrix of $\mathbf{C}$ consisting of block rows  indexed by $j\in U_i$, is invertible and thus has  rank $\rf\ell$. 
    Removing rows indexed by \( D_{(i-1)\kf + j} \) for every \( j \in U_i \) yields a submatrix of \( \mathbf{C}^{(i)} \).
    Let
   \[
   E_{i,a}:=\overline{D_{(i-1)\kf+u_{i,a}}},\quad a\in[\rf].
   \]
   The resulting submatrix is
   $$\begin{bmatrix} \mathbf{C}_{u_{i,1},1}[E_{i,1};:]& \cdots& \mathbf{C}_{u_{i,1},\rf}[E_{i,1};:]\\ \vdots&\ddots&\vdots\\ \mathbf{C}_{u_{i,\rf},1}[E_{i,\rf};:]& \cdots& \mathbf{C}_{u_{i,\rf},\rf}[E_{i,\rf};:] \end{bmatrix}.$$
    We have
    $$\rank(\mathbf{C}^{(i)})\ge \rf\ell-\left(\sum\limits_{j\in U_i} \beta_{(i-1)\kf+j}\right).$$
    Summing over all $\binom{k^F}{r^F}$ such subsets $U_i$, each $j\in[k^F]$ appears in exactly $\binom{k^F-1}{r^F-1}$ of them, giving
   \[
\begin{aligned}
\binom{k^F}{r^F}\operatorname{rank}(\mathbf{C}^{(i)})
&\ge
\binom{k^F}{r^F} r^F\ell  \\
&\quad
-
\binom{k^F-1}{r^F-1}
\sum_{j\in[k^F]} \beta_{(i-1)k^F+j}.
\end{aligned}
\]
   Then, by $\binom{k^F}{r^F}/\binom{k^F-1}{r^F-1}=k^F/r^F$, the above inequality can be simplified as
   $$\kf \rank(\mathbf{C}^{(i)})\ge\kf \rf\ell- \rf\sum\limits_{j\in [{\kf}]} \beta_{(i-1)\kf+j}. $$
   By summing over all $i\in[\lambda]$, this yields 
   \begin{align*} 
        \sum\limits_{i=1}^{\lambda}\kf\rank(\mathbf{C}^{(i)})
        &\ge\lambda\kf\rf\ell-\rf\sum\limits_{i=1}^{\lambda}\sum\limits_{j\in [\kf]} \beta_{(i-1)\kf+j}\\
        &\ge\lambda\kf\rf\ell-\rf\sum\limits_{j\in [\ki]} \beta_{j}.
    \end{align*}
    Combining this with~\eqref{bound1_total_rank}, we obtain the following inequality: 
    \begin{equation}\label{bound1_condition1}
        \lambda\kf\rf\ell-\rf\sum\limits_{i=1}^{\ki} \beta_{i}{\le}\kf\sum\limits_{i=\ki+1}^{\nii} \beta_i.
    \end{equation}
     
    By~\eqref{character_eq}, 
    the subspace $\langle\tilde{\mathbf{C}}\rangle$  can be expanded to column space $\langle\tilde{\mathbf{B}}\rangle$ by adding $\rank(\tilde{\mathbf{B}})-\rank(\tilde{\mathbf{C}})$ column vectors. For each $i\in[\lambda]$, we denote $\mathbf{B}^{(i)}$ as the submatrix of $\tilde{\mathbf{B}}$ obtained by restricting $\tilde{\mathbf{B}}$ to block rows indexed by $(i-1)k^F+1$ through $i k^F$. Define
   \[
F_{i,a}:=\overline{D_{(i-1)\kf+a}},\;
s_{i,a}:=(i-1)\kf+a,\; a\in[\kf].
\]
 $$Q_b:=D_{\ki+b},\; b\in[\ri].$$
Then $\mathbf{B}^{(i)}$ can be written as
    \begin{equation}\label{def:Bi}
\begin{bmatrix}
\mathbf{B}_{s_{i,1},1}[F_{i,1}; Q_1] 
& \cdots & 
\mathbf{B}_{s_{i,1},\ri}[F_{i,1}; Q_{\ri}] \\
\vdots & \ddots & \vdots \\
\mathbf{B}_{s_{i,\kf},1}[F_{i,\kf}; Q_1] 
& \cdots & 
\mathbf{B}_{s_{i,\kf},\ri}[F_{i,\kf}; Q_{\ri}]
\end{bmatrix}.
\end{equation}
    Then, by \cref{character_lemma}, we have 
    \begin{equation}\label{condition2}
        \rank(\mathbf{B}^{(i)})-\rank(\mathbf{C}^{(i)})\le\rank(\tilde{\mathbf{B}})-\rank(\tilde{\mathbf{C}}).
    \end{equation}
     Summing  over all $i\in[\lambda]$, this implies that
    \begin{align}
    \begin{split}\label{sum_main_ineq}
        \sum_{i=1}^{\lambda}\rank(\mathbf{B}^{(i)})\le\lambda\rank(\tilde{\mathbf{B}})-(\lambda-1)\rank(\tilde{\mathbf{C}}).
    \end{split}
    \end{align}

    Next, we provide lower bounds on the ranks of matrices $\mathbf{B}^{(i)}$ and $\tilde{\mathbf{C}}$, respectively.

    We start with bounding $\rank(\mathbf{B}^{(i)})$. 
    Since $\mathcal{C}^I$ is an MDS array code, every $r^I\times r^I$ block submatrix of $\mathbf{B}$ is nonsingular.
    Noting that $r^I\le k^F$, for every $H=\{h_1,\cdots,h_{r^I}\}\subseteq[\kf]$,
    \begin{equation*}
\begin{aligned}
&\rank
\begin{bmatrix}
\mathbf{B}_{s_{i,h_1},1}[:;Q_1]
& \cdots &
\mathbf{B}_{s_{i,h_1},\ri}[:;Q_{\ri}] \\
\vdots & \ddots & \vdots \\
\mathbf{B}_{s_{i,h_{\ri}},1}[:;Q_1]
& \cdots &
\mathbf{B}_{s_{i,h_{\ri}},\ri}[:;Q_{\ri}]
\end{bmatrix} \\
&\qquad =
\sum_{j=\ki+1}^{\nii}\beta_j .
\end{aligned}
\end{equation*}
    Deleting rows with index in $D_{(i-1)\kf+h_1},\cdots,D_{(i-1)\kf+h_{r^{I}}}$, we have that \(\rank(\mathbf{B}^{(i)})\) is at least
    \begin{align*}
&\rank
{\setlength{\arraycolsep}{1.5pt}
\begin{bmatrix}
    \mathbf{B}_{s_{i,h_1},1}[F_{i,h_1};Q_1]
& \cdots &
\mathbf{B}_{s_{i,h_1},\ri}[F_{i,h_1};Q_{\ri}] \\
\vdots & \ddots & \vdots \\
\mathbf{B}_{s_{i,h_{\ri}},1}[F_{i,h_{\ri}};Q_1]
& \cdots &
\mathbf{B}_{s_{i,h_{\ri}},\ri}[F_{i,h_{\ri}};Q_{\ri}]
\end{bmatrix}}\\
&\ge
\sum_{j=k^I+1}^{n^I}\beta_j
-
\sum_{j\in H}\beta_{(i-1)k^F+j}.
\end{align*}
    Then, by summing over all possible $H=\{h_1,\cdots,h_{\ri}\}\subseteq[\kf]$, the above inequality implies that
    \begin{align*}
        &\sum\limits_{H\subseteq[\kf],\, |H|=\ri}\rank(\mathbf{B}^{(i)})\\&\ge\sum\limits_{H\subseteq[\kf],\, |H|=\ri}(\sum\limits_{j=\ki+1}^{\nii}\beta_j-\sum\limits_{j\in H} \beta_{(i-1)\kf+j})\\
        &=\binom{\kf}{\ri}\sum\limits_{j=\ki+1}^{\nii}\beta_j-\binom{\kf-1}{\ri-1}\sum\limits_{j=(i-1)\kf+1}^{i\kf} \beta_j .
    \end{align*}
    This leads to
    \begin{equation}\label{B_i}
       \rank(\mathbf{B}^{(i)})\ge \sum\limits_{j=\ki+1}^{\nii}\beta_j-\frac{\ri}{\kf}\sum\limits_{j=(i-1)\kf+1}^{i\kf} \beta_j.
    \end{equation}
    Summing over all possible $i\in[\lambda]$, we have
    \begin{equation}\label{sum_2}
        \sum_{i=1}^{\lambda}\rank(\mathbf{B}^{(i)}) 
    \ge \lambda\sum_{j=k^I+1}^{n^I}\beta_j 
      - \frac{r^I}{k^F}\sum_{j=1}^{k^I}\beta_j. 
    \end{equation}
    
    We next bound $\rank(\tilde{\mathbf{C}})$. Noting that $\rf\le\kf$, for each $i\in[\lambda]$ we can choose a set $I_i=\{g_{i,1},\ldots,g_{i,\rf}\}\subseteq[\kf]$ of size $\rf$. Since $\cf$ is an MDS code, the corresponding $\rf\times\rf$ block submatrix of $\mathbf{C}$ (formed by the block rows indexed by $I_i$ and the block columns $[r^F]$) is invertible, and hence has rank $\rf\ell$. Therefore, we have  that \(\rank(\mathbf{C}^{(i)})\) is at least
   \begin{align*}
&\rank
\begin{bmatrix}
\mathbf{C}_{g_{i,1},1}[F_{i,g_{i,1}};:]
& \cdots &
\mathbf{C}_{g_{i,1},\rf}[F_{i,g_{i,1}};:] \\
\vdots & \ddots & \vdots \\
\mathbf{C}_{g_{i,\rf},1}[F_{i,g_{i,\rf}};:]
& \cdots &
\mathbf{C}_{g_{i,\rf},\rf}[F_{i,g_{i,\rf}};:]
\end{bmatrix} \\
&\ge
\rf\ell-\sum_{j\in I_i}\beta_{(i-1)\kf+j}.
\end{align*} 
   where the second inequality holds because removing $\sum_{j\in I_i}\beta_{(i-1)\kf+j}$ rows can decrease the rank by at most that number.
   This implies that
   \begin{align*}
        \rank(\tilde{\mathbf{C}})&=\sum_{i=1}^{\lambda} \rank(\mathbf{C}^{(i)})\\
        &\ge \sum_{i=1}^{\lambda}(\rf\ell-\sum\limits_{j\in I_i} \beta_{(i-1)\kf+j})\\
        &=\lambda\rf\ell-\sum_{i=1}^{\lambda}
\sum_{j\in I_i}\beta_{(i-1)\kf+j}.
   \end{align*}
  By summing both sides of the above inequality over all possible $I_1,I_2,\ldots,I_{\lambda}\subseteq[\kf]$, we have
   \begin{align*}
&\binom{\kf}{\rf}^{\lambda}\rank(\tilde{\mathbf{C}})
=
\binom{\kf}{\rf}^{\lambda}
\sum_{i=1}^{\lambda}\rank(\mathbf{C}^{(i)}) \\
\ge
\lambda&\binom{\kf}{\rf}^{\lambda}\rf\ell 
-
\sum_{\substack{
I_1,\ldots,I_\lambda\subseteq[\kf]\\
|I_1|=\cdots=|I_\lambda|=\rf
}}
\sum_{i=1}^{\lambda}
\sum_{j\in I_i}
\beta_{(i-1)\kf+j} \\
=
\lambda&\binom{\kf}{\rf}^{\lambda}\rf\ell 
-\binom{\kf-1}{\rf-1}
\binom{\kf}{\rf}^{\lambda-1}
\sum_{j=1}^{\ki}\beta_j,
\end{align*}
   where the last equality follows since each $j\in [k^F]$ appears in exactly
   $\binom{k^F-1}{r^F-1}$ different subsets of $[k^F]$ with size $r^F$. This further implies that
   \begin{equation}\label{c_tilde}
       \rank(\tilde{\mathbf{C}})\ge\lambda\rf\ell-\frac{\rf}{\kf}\sum_{i=1}^{\ki}\beta_i.
   \end{equation}
   
   Finally, since $\rank(\tilde{\mathbf{B}})=\sum\limits_{i=\ki+1}^{\nii}\beta_i$, then by~\eqref{sum_main_ineq},\eqref{sum_2} and \eqref{c_tilde}, it holds that:
   \begin{equation}\label{bound1_condition2}
       \sum_{i=1}^{\ki}\beta_i\ge\frac{\lambda k^F(\lambda-1)\rf\ell}{(\lambda-1)\rf+\ri}.
   \end{equation}
    Now, based on the linear constraints~\eqref{bound1_condition1} and~\eqref{bound1_condition2}, we have the following linear optimization problem:
    \begin{equation*}
    \begin{array}{rl}
    \text{minimize} & R = \sum\limits_{i=1}^{\nii} \beta_i \\
    \text{subject to} & \eqref{bound1_condition1},~\eqref{bound1_condition2}, \\
                      & 0 \le \beta_i \le \ell, \quad \text{for } i \in [\nii].
    \end{array}
    \end{equation*}
    Then, this LP problem can be easily solved to obtain the desired lower bound.
\end{proof}
 \begin{thm}\label{bound_3}
    For any stable MDS $[\nii,\ki=\lambda\kf;\nf,\kf;\ell]$ convertible code satisfying $\rf<\kf<\ri$, the read bandwidth cost $R$ is lower bounded by:
    \begin{equation}
    \begin{cases} 
        \lambda\rf\ell, & \text{if } \ki \le \ri; \\[8pt]
        
        \displaystyle \frac{\lambda^2(\kf)^2 \rf\ell}{\kf \ri - \rf \ri + \lambda\kf \rf}, & \text{if } \ki > \ri, \Lambda \ge 0; \\[10pt]
        
        \displaystyle \ri\ell + \frac{ \lambda\kf[(\lambda-1)\rf-( \ri - \kf)]\ell }{(\lambda-1)\rf+\kf}, & \text{if } \ri < \Gamma, \Lambda \le 0; \\[10pt]
        
        \displaystyle \frac{ (\lambda-1)\rf\ri\ell }{\ri-\kf}, & \begin{aligned}
            & \text{if } \ki > \ri \ge \Gamma \\ 
            & \text{and } \Lambda \le 0,
        \end{aligned}
    \end{cases}
    \end{equation}
    where $\Lambda = \lambda(\kf)^2-(\lambda-1)(\kf-\rf)\ri$, $\Gamma = (\lambda-1)\rf+\kf$.
\end{thm}
\begin{proof}
Let $(\ci,\cf)$ along with the conversion procedure $\mathcal{T}$ form a stable MDS convertible code with $\rf<\kf<\ri$, achieving the minimum bandwidth cost. 
Since $r^F\leq k^F$, by similar arguments as those in the proof of \cref{bound1}, constraints~\eqref{bound1_condition1} and~\eqref{condition2} remain valid. 
Moreover, the bound on the rank of $\tilde{\mathbf{C}}$ in~\eqref{c_tilde} also holds, since its derivation only relies on the condition $r^F \le k^F$.
However, due to $k^F \le r^I$ the expression of $\rank(\mathbf{B}^{(i)})$ changes. 

For fixed \(i\in[\lambda]\), define
\[
s_{i,a}:=(i-1)\kf+a,\ a\in[\kf],
\quad
Q_b:=D_{\ki+b},\ b\in[\ri].
\]
For any subset $J=\{j_1, \dots, j_{k^F}\} \subseteq [r^I]$, consider the $k^F\times k^F$ block submatrix
$$
\mathbf{M}_J = 
\begin{bmatrix}
       \mathbf{B}_{s_{i,1}, j_1} & \cdots & \mathbf{B}_{s_{i,1}, j_{k^F}}\\
               \vdots & \ddots & \vdots \\
                 \mathbf{B}_{s_{i,\kf}, j_1} & \cdots & \mathbf{B}_{s_{i,\kf}, j_{k^F}}
\end{bmatrix}.
$$
Since $\mathcal{C}^I$ is an MDS array code and $k^F\le r^I$, $\mathbf{M}_J$ is nonsingular, hence its  columns are linearly independent.
Let $\widetilde{\mathbf{M}}_J$ denote the matrix obtained by selecting the  columns indexed by $D_{k^I+j_m}$ from each block column $j_m$ of $\mathbf{M}_J$ for $m\in[k^F]$:
$$
\widetilde{\mathbf{M}}_J = 
\begin{bmatrix}
       \mathbf{B}_{s_{i,1}, j_1}[:; Q_{j_1}] & \cdots & \mathbf{B}_{s_{i,1}, j_{k^F}}[:; Q_{j_{k^F}}]\\
               \vdots & \ddots & \vdots \\
                 \mathbf{B}_{s_{i,\kf}, j_1}[:; Q_{j_1}] & \cdots & \mathbf{B}_{s_{i,\kf}, j_{k^F}}[:; Q_{j_{k^F}}]
\end{bmatrix}.
$$
Since $\mathbf{M}_J$ has full column rank, it follows that $\widetilde{\mathbf{M}}_J$ has full column rank and
$\rank(\widetilde{\mathbf{M}}_J) = \sum_{u \in J}\beta_{k^I+u}$.
Then we obtain
\begin{align*}
&\rank
\begin{bmatrix}
       \mathbf{B}_{s_{i,1},1}[:; Q_1] & \cdots & \mathbf{B}_{s_{i,1}, r^I}[:; Q_{\ri}] \\
               \vdots & \ddots & \vdots \\
                 \mathbf{B}_{s_{i,\kf},1}[:; Q_1] & \cdots & \mathbf{B}_{s_{i,\kf}, r^I}[:; Q_{\ri}]
\end{bmatrix}\\
&\ge \rank(\widetilde{\mathbf{M}}_J)
= \sum_{u \in J}\beta_{k^I+u}.
\end{align*}
Because this inequality holds for any subset $J$ of size $k^F$, it follows:
\begin{align*}
&\rank
\begin{bmatrix}
       \mathbf{B}_{s_{i,1},1}[:; Q_1] & \cdots & \mathbf{B}_{s_{i,1}, r^I}[:; Q_{\ri}] \\
               \vdots & \ddots & \vdots \\
                 \mathbf{B}_{s_{i,\kf},1}[:; Q_1] & \cdots & \mathbf{B}_{s_{i,\kf}, r^I}[:; Q_{\ri}]
\end{bmatrix}\\
&\ge\max_{J\subseteq[r^I],\,|J|=k^F}\left\{\sum_{u\in J}\beta_{k^I+u}\right\}
\ge \frac{k^F}{r^I}\sum_{j=k^I+1}^{n^I}\beta_j.
\end{align*}

  Deleting rows in $D_{(i-1)\kf+1},\cdots,D_{i\kf}$, and by the definition of $\mathbf{B}^{(i)}$ in \eqref{def:Bi}, we have
   \begin{align}\label{B_i2}
        \rank(\mathbf{B}^{(i)})&\ge\frac{\kf}{\ri}\sum_{j=\ki+1}^{\nii}\beta_j-\sum_{j=(i-1)\kf+1}^{i\kf}\beta_j.
   \end{align}
   Summing over all possible $i\in[\lambda]$, we have
   \begin{equation}
     \sum_{i=1}^{\lambda}\rank(\mathbf{B}^{(i)})\ge\frac{\lambda\kf}{\ri}\sum_{j=\ki+1}^{\nii}\beta_j-\sum_{j=1}^{\ki}\beta_j.
\end{equation}
 By~\eqref{sum_main_ineq} and~\eqref{c_tilde}, we have
 \begin{align}
 \begin{split}\label{bound2_condition2}
     &\frac{(\lambda-1)\rf+\kf}{\kf}\sum_{i=1}^{\ki}\beta_i+\frac{\lambda(\ri-\kf)}{\ri}\sum_{i=\ki+1}^{\nii}\beta_i\\&\quad\ge\lambda(\lambda-1)\rf\ell.
 \end{split}
 \end{align}  

 Next, to obtain the minimum value of the read bandwidth cost $R$, we consider the following optimization problem:
 \begin{equation*}
    \begin{array}{rl}
        \text{minimize} & R = \sum\limits_{j\in [\nii]}\beta_j \\
        \text{subject to} & \eqref{bound1_condition1},\eqref{bound2_condition2}  \\
                          & 0\le\beta_i\le\ell,\ \text{for}\  i\in[\nii].
    \end{array}
    \end{equation*}
    Let $x=\sum\limits_{i=1}^{\ki}\beta_i$ and $y=\sum\limits_{i=\ki+1}^{\nii}\beta_i$.
    Then the optimization problem can be rewritten as
   \begin{align*}
    \min \;& R = x + y \\[1mm]
    \text{subject to} \;& \rf x + \kf y \ge \lambda \kf\rf \ell, \\[1mm]
                  & \frac{(\lambda-1) \rf + \kf}{\kf} x + \frac{\lambda (\ri - \kf)}{\ri} y \ge \lambda (\lambda-1) \rf \ell, \\[1mm]
                  & 0 \le x \le\lambda\kf  \ell, \quad 0 \le y \le \ri \ell.
\end{align*}
To solve this linear program, we  analyze the feasible region in the $x$-$y$ plane. 
\begin{itemize}
    \item When $\ki \le \ri$, the second constraint does not further restrict the feasible region determined by the first constraint, as illustrated in Fig.~\ref{fig:shaded_regions}(a).
Hence, the optimal solution is attained at $(0,\, \lambda \rf \ell)$, yielding
\[
R = \lambda \rf \ell.
\]
\item When $\ki > \ri$ and $\lambda(\kf)^2-(\lambda-1)(\kf-\rf)\ri\ge 0$,  the optimal solution lies at the intersection of the equality boundaries corresponding to the two constraints as in Fig.~\ref{fig:shaded_regions}(b).
The coordinates of the optimal point $(p,q)$ are
$$\left( \frac{ \lambda\kf \rf (\lambda\kf - \ri)\ell}{\kf\ri - \rf\ri + \lambda\kf\rf},\;
       \frac{ \lambda\kf\rf\ri\ell}{\kf\ri - \rf\ri + \lambda\kf\rf} \right),$$
with corresponding optimal value
$$R= \frac{\lambda^2(\kf)^2\rf\ell }{\kf\ri- \rf\ri + \lambda\kf\rf}.$$
   \item When $\ri < (\lambda-1)\rf + \kf$ and $\lambda(\kf)^2 - (\lambda-1)(\kf-\rf)\ri \le 0$, the slopes of the equality boundaries corresponding to the two constraints are both no less than $-1$. 
   In this case, the optimal solution occurs at the intersection of the second constraint boundary and the line $y = \ri \ell$, as illustrated in Fig.~\ref{fig:shaded_regions}(c).
  The coordinates of the optimal point are
  \[
  \left( 
\frac{ \lambda\kf[(\lambda-1)\rf-( \ri - \kf)]\ell }{(\lambda-1)\rf+\kf },\;
\ri\ell
\right),
 \]
with the corresponding optimal value
\[
R = \ri\ell+ 
\frac{ \lambda\kf[(\lambda-1)\rf-( \ri - \kf)]\ell }{(\lambda-1)\rf+\kf}.
\]

\item When $\ki> \ri\ge (\lambda-1)\rf+\kf$ and $\lambda(\kf)^2-(\lambda-1)(\kf-\rf)\ri\le0$, the optimal solution occurs at the intersection of the second constraint boundary and the line $x = 0$, as illustrated in~ Fig.~\ref{fig:shaded_regions}(d). 
  The coordinates of the optimal point are
  \[
  \left( 0,\;
\frac{ (\lambda-1)\rf\ri\ell }{\ri-\kf}
\right),
 \]
with the corresponding optimal value
\(R=\frac{ (\lambda-1)\rf\ri\ell }{\ri-\kf}.\)
\end{itemize}

This gives the desired result.
\end{proof}

\begin{remark}
As established in the linear programming analysis of the proof above, the optimal read bandwidth is obtained by minimizing $R = x + y$, where $x$ and $y$ represent the amount of data downloaded from the unchanged symbols and the retired symbols, respectively. 

For the first and fourth cases of Theorem~\ref{bound_3}, the optimal solution occurs at $x=0$. This corresponds to the situation where all downloaded data come solely from the retired symbols, and no data are downloaded from the unchanged symbols. 

For the second and third cases of Theorem~\ref{bound_3}, the optimal coordinates yield $x > 0$ and $y > 0$, meaning achieving the bound requires downloading data from both the unchanged and retired symbols. Moreover, in the third case, the optimal solution specifically lies on the upper bound $y = r^I \ell$, which implies that all data stored in the retired symbols must be completely accessed. In particular, the second case is shown to be optimal for certain parameters by the example provided in Appendix \ref{ppendix_A}. 
\end{remark}

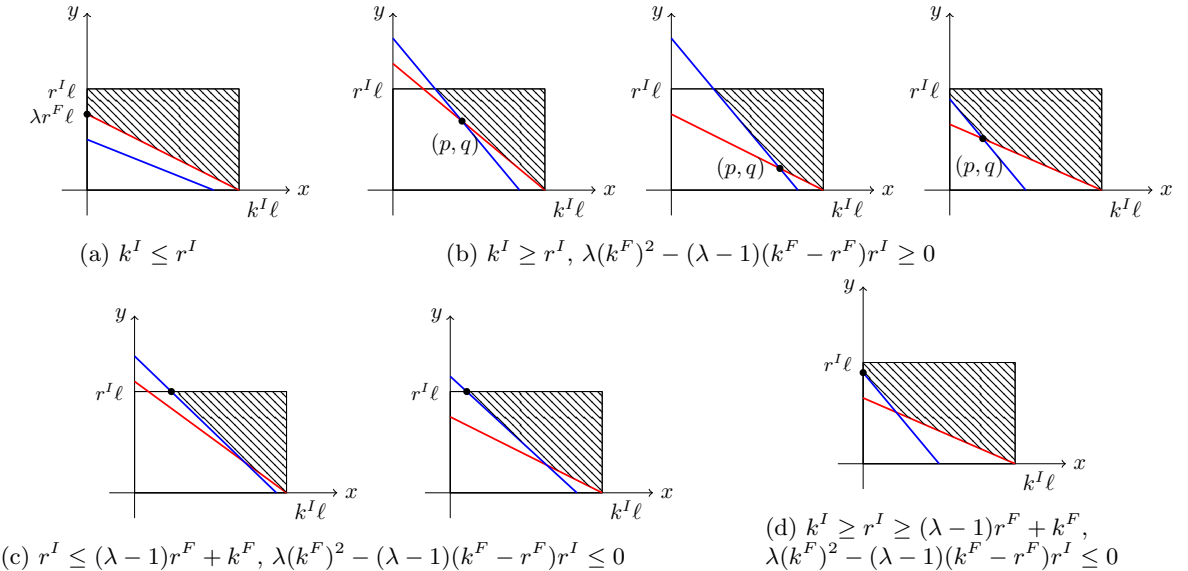
\begin{figure*}[!t]
\centering

\newlength{\LPplotwidth}
\newlength{\LPgap}
\setlength{\LPplotwidth}{0.145\textwidth}
\setlength{\LPgap}{0.018\textwidth}

\newcommand{\LPtikz}[1]{%
\resizebox{\LPplotwidth}{!}{#1}%
}

\begin{tabular}{@{}c@{\hspace{\LPgap}}c@{\hspace{\LPgap}}c@{\hspace{\LPgap}}c@{}}

\LPtikz{%
\begin{tikzpicture}[scale=0.8]
    \draw[->] (-0.5,0) -- (4,0) node[right] {$x$};
    \draw[->] (0,-0.5) -- (0,3.5) node[left] {$y$};

    \draw (0,0) rectangle (3,2);
    \draw[name path=rect_left] (0,0) -- (0,2);
    \draw[name path=rect_right] (3,0) -- (3,2);
    \draw[name path=rect_up] (0,2) -- (3,2);
    \draw[name path=rect_down] (0,0) -- (3,0);

    \draw[red, thick, name path=red_line] (0,1.5) -- (3,0);
    \node[below right] at (3,0) {$k^I\ell$};
    \draw[blue, thick, name path=blue_line] (0,1) -- (2.5,0);

    \path[name intersections={of=red_line and rect_left, by=red_left}];

    \fill[pattern=north west lines]
        (red_left) -- (0,2) -- (3,2) -- (3,0) -- cycle;

    \fill (0,1.5) circle (2pt);
    \node at (-0.5,2) {\small$r^I\ell$};
    \node at (-0.7,1.5) {\small$\lambda\rf\ell$};
\end{tikzpicture}
}
&
\LPtikz{%
\begin{tikzpicture}[scale=0.8]
    \draw[->] (-0.5,0) -- (4,0) node[right] {$x$};
    \draw[->] (0,-0.5) -- (0,3.5) node[left] {$y$};

    \draw (0,0) rectangle (3,2);
    \draw[name path=rect_left] (0,0) -- (0,2);
    \draw[name path=rect_right] (3,0) -- (3,2);
    \draw[name path=rect_up] (0,2) -- (3,2);
    \draw[name path=rect_down] (0,0) -- (3,0);

    \draw[red, thick, name path=red_line] (0,2.5) -- (3,0);
    \node[below right] at (3,0) {$k^I\ell$};
    \draw[blue, thick, name path=blue_line] (0,3) -- (2.5,0);

    \path[name intersections={of=red_line and blue_line, by=line_intersect}];
    \path[name intersections={of=blue_line and rect_up, by=blue_up}];
    \path[name intersections={of=red_line and rect_right, by=red_right}];

    \fill[pattern=north west lines]
        (blue_up) -- (line_intersect) -- (red_right) -- (3,2) -- cycle;

    \node at (-0.5,2) {$r^I\ell$};
    \fill (line_intersect) circle (2pt);
    \node[left=3pt, below=2pt] at (line_intersect) {$(p,q)$};
\end{tikzpicture}
}
&
\LPtikz{%
\begin{tikzpicture}[scale=0.8]
    \draw[->] (-0.5,0) -- (4,0) node[right] {$x$};
    \draw[->] (0,-0.5) -- (0,3.5) node[left] {$y$};

    \draw (0,0) rectangle (3,2);
    \draw[name path=rect_left] (0,0) -- (0,2);
    \draw[name path=rect_right] (3,0) -- (3,2);
    \draw[name path=rect_up] (0,2) -- (3,2);
    \draw[name path=rect_down] (0,0) -- (3,0);

    \draw[red, thick, name path=red_line] (0,1.5) -- (3,0);
    \node[below right] at (3,0) {$k^I\ell$};
    \draw[blue, thick, name path=blue_line] (0,3) -- (2.5,0);

    \path[name intersections={of=red_line and blue_line, by=line_intersect}];
    \path[name intersections={of=blue_line and rect_up, by=blue_up}];
    \path[name intersections={of=red_line and rect_right, by=red_right}];

    \fill[pattern=north west lines]
        (blue_up) -- (line_intersect) -- (red_right) -- (3,2) -- cycle;

    \node at (-0.5,2) {$r^I\ell$};
    \fill (line_intersect) circle (2pt);
    \node[left=3pt] at (line_intersect) {$(p,q)$};
\end{tikzpicture}
}
&
\LPtikz{%
\begin{tikzpicture}[scale=0.8]
    \draw[->] (-0.5,0) -- (4,0) node[right] {$x$};
    \draw[->] (0,-0.5) -- (0,3.5) node[left] {$y$};

    \draw (0,0) rectangle (3,2);
    \draw[name path=rect_left] (0,0) -- (0,2);
    \draw[name path=rect_right] (3,0) -- (3,2);
    \draw[name path=rect_up] (0,2) -- (3,2);
    \draw[name path=rect_down] (0,0) -- (3,0);

    \draw[red, thick, name path=red_line] (0,1.3) -- (3,0);
    \node[below right] at (3,0) {$k^I\ell$};
    \draw[blue, thick, name path=blue_line] (0,1.8) -- (1.5,0);

    \path[name intersections={of=red_line and blue_line, by=line_intersect}];
    \path[name intersections={of=blue_line and rect_left, by=blue_left}];
    \path[name intersections={of=red_line and rect_right, by=red_right}];

    \fill[pattern=north west lines]
        (blue_left) -- (line_intersect) -- (red_right) -- (3,2) -- (0,2) -- cycle;

    \node at (-0.5,2) {$r^I\ell$};
    \fill (line_intersect) circle (2pt);
    \node[left=1.5pt, below=4pt] at (line_intersect) {$(p,q)$};
\end{tikzpicture}
}
\\[-0.2em]
\small (a) $\ki\le\ri$
&
\multicolumn{3}{c}{\small (b) $\ki\ge\ri$, $\lambda(\kf)^2-(\lambda-1)(\kf-\rf)\ri\ge0$}
\end{tabular}

\vspace{0.3em}

\begin{tabular}{@{}c@{\hspace{\LPgap}}c@{\hspace{\LPgap}}c@{}}

\LPtikz{%
\begin{tikzpicture}[scale=0.8]
    \draw[->] (-0.5,0) -- (4,0) node[right] {$x$};
    \draw[->] (0,-0.5) -- (0,3.5) node[left] {$y$};

    \draw (0,0) rectangle (3,2);
    \draw[name path=rect_left] (0,0) -- (0,2);
    \draw[name path=rect_right] (3,0) -- (3,2);
    \draw[name path=rect_up] (0,2) -- (3,2);
    \draw[name path=rect_down] (0,0) -- (3,0);

    \draw[red, thick, name path=red_line] (0,2.2) -- (3,0);
    \node[below right] at (3,0) {$k^I\ell$};
    \draw[blue, thick, name path=blue_line] (0,2.7) -- (2.8,0);

    \path[name intersections={of=red_line and blue_line, by=line_intersect}];
    \path[name intersections={of=blue_line and rect_up, by=blue_up}];
    \path[name intersections={of=red_line and rect_right, by=red_right}];

    \fill[pattern=north west lines]
        (blue_up) -- (line_intersect) -- (3,0) -- (3,2) -- cycle;

    \fill (blue_up) circle (2pt);
    \node at (-0.5,2) {\small$r^I\ell$};
\end{tikzpicture}
}
&
\LPtikz{%
\begin{tikzpicture}[scale=0.8]
    \draw[->] (-0.5,0) -- (4,0) node[right] {$x$};
    \draw[->] (0,-0.5) -- (0,3.5) node[left] {$y$};

    \draw (0,0) rectangle (3,2);
    \draw[name path=rect_left] (0,0) -- (0,2);
    \draw[name path=rect_right] (3,0) -- (3,2);
    \draw[name path=rect_up] (0,2) -- (3,2);
    \draw[name path=rect_down] (0,0) -- (3,0);

    \draw[red, thick, name path=red_line] (0,1.5) -- (3,0);
    \node[below right] at (3,0) {$k^I\ell$};
    \draw[blue, thick, name path=blue_line] (0,2.3) -- (2.5,0);

    \path[name intersections={of=red_line and blue_line, by=line_intersect}];
    \path[name intersections={of=blue_line and rect_up, by=blue_up}];
    \path[name intersections={of=red_line and rect_right, by=red_right}];

    \fill[pattern=north west lines]
        (blue_up) -- (line_intersect) -- (3,0) -- (3,2) -- cycle;

    \fill (blue_up) circle (2pt);
    \node at (-0.5,2) {\small$r^I\ell$};
\end{tikzpicture}
}
&
\LPtikz{%
\begin{tikzpicture}[scale=0.8]
    \draw[->] (-0.5,0) -- (4,0) node[right] {$x$};
    \draw[->] (0,-0.5) -- (0,3.5) node[left] {$y$};

    \draw (0,0) rectangle (3,2);
    \draw[name path=rect_left] (0,0) -- (0,2);
    \draw[name path=rect_right] (3,0) -- (3,2);
    \draw[name path=rect_up] (0,2) -- (3,2);
    \draw[name path=rect_down] (0,0) -- (3,0);

    \draw[red, thick, name path=red_line] (0,1.3) -- (3,0);
    \node[below right] at (3,0) {$k^I\ell$};
    \draw[blue, thick, name path=blue_line] (0,1.8) -- (1.5,0);

    \path[name intersections={of=red_line and blue_line, by=line_intersect}];
    \path[name intersections={of=blue_line and rect_left, by=blue_left}];
    \path[name intersections={of=red_line and rect_right, by=red_right}];

    \fill[pattern=north west lines]
        (blue_left) -- (line_intersect) -- (3,0) -- (3,2) -- (0,2) -- cycle;

    \fill (blue_left) circle (2pt);
    \node at (-0.5,2) {\small$r^I\ell$};
\end{tikzpicture}
}
\\[-0.2em]
\multicolumn{2}{c}{\small (c) $\ri\le(\lambda-1)\rf+\kf$, $\lambda(\kf)^2-(\lambda-1)(\kf-\rf)\ri\le0$}
&
\small (d) $\ki\ge\ri\ge(\lambda-1)\rf+\kf$
\end{tabular}

\vspace{-0.3em}

\caption{Feasible regions of the linear program. The red line corresponds to the boundary of the first constraint, and the blue line corresponds to the boundary of the second constraint.}
\label{fig:shaded_regions}
\end{figure*}
\section{Conclusion}\label{conclusion}
In this paper, we introduce a linear-algebraic framework for analyzing the conversion bandwidth of MDS convertible codes in the split regime. Using this, we  derive closed-form lower bounds on the read bandwidth during conversion.  The key insight is that the conversion imposes a subspace inclusion  relation between certain restricted columns of the generator matrices of the initial and final codes. This inclusion naturally leads to a set of linear programming constraints  whose optimal solution yields the desired lower bounds on bandwidth cost.

Next we compare our bounds  with the following bound proposed by  Maturana and  Rashmi  in ~\cite{Maturana2022BandwidthCO}:
\begin{equation}\label{thm_4}
    R \ge
\begin{cases}
\lambda \kf \ell-\ri\ell\max\{\frac{\kf}{\rf}-1,0\} & \text{if } \ri \le \lambda \rf,\\[2mm]
\lambda \min\{\rf,\kf\}  \ell, & \text{if } \ri > \lambda \rf.
\end{cases}
\end{equation}
\begin{itemize}
    \item For $r^F \ge k^F$, our bound in~\cref{bound0} equals $\ki\ell$, which matches the expression in~\eqref{thm_4}. This is tight by definition, as $\ki\ell$ corresponds exactly to the total number of message subsymbols.
    \item For $\ri\le r^F\le k^F$, our bound in~\cref{bound1} is $R \geq \lambda \kf \ell-\frac{(\kf-\rf)\ri}{\rf}\ell$. 
    This coincides with the expression in~\eqref{thm_4} and is known to be tight, as  it is achieved  by the construction of Maturana and Rashmi for this parameter range.~\cite{Maturana2022BandwidthCO}
    \item For $r^F \le r^I \le k^F$, our bound in~\cref{bound1} is \begin{equation*}
        R \geq \lambda\rf\ell \frac{(\lambda-1)\kf+\ri}{(\lambda-1)\rf+\ri},
    \end{equation*}
    while the bound in~\eqref{thm_4} is 
    \begin{equation*}
        R \ge
\begin{cases}
\lambda \kf \ell-\frac{(\kf-\rf)\ri}{\rf}\ell & \text{if } \ri \le \lambda \rf,\\[2mm]
\lambda \rf \ell, & \text{if } \ri > \lambda \rf.
\end{cases}
    \end{equation*}
If $\ri \le \lambda \rf$, we have
\begin{align*}
   &\left(\lambda \kf \ell-\frac{(\kf-\rf)\ri}{\rf}\ell\right)/\left(\lambda\rf\ell \frac{(\lambda-1)\kf+\ri}{(\lambda-1)\rf+\ri}\right)\\&=1-\frac{\ri(\kf-\rf)(\ri-\rf)}{\lambda(\rf)^2\left((\lambda-1)\kf+\ri\right)}\\&\le1,
\end{align*}
where the inequality holds since each multiplicative factor in the product is nonnegative. \\
If $\ri > \lambda \rf$, we have
\begin{align*}
 &(\lambda \rf \ell)/\left(\lambda\rf\ell \frac{(\lambda-1)\kf+\ri}{(\lambda-1)\rf+\ri}\right)\\
 &=\frac{(\lambda-1)\rf+\ri}{(\lambda-1)\kf+\ri}\le1.
\end{align*}
In total, our bound is better than  the bound in~\eqref{thm_4}. Moreover, the lower bound given in~\cref{bound1} with $r^F \le r^I \le k^F$ is attainable  by the construction of Maturana and Rashmi in ~\cite{Maturana2022BandwidthCO}, establishing its tightness for this range.

    \item  For $r^F < k^F\le \ki\le \ri$, our bound in~\cref{bound_3} is $ \lambda r^F\ell$, which  agrees with~\eqref{thm_4}.
    \item For $r^F < k^F\le \ri< \ki$, we have the following comparisons.
    \begin{enumerate}
        \item If $\ki > \ri$, $\lambda(\kf)^2-(\lambda-1)(\kf-\rf)\ri\ge 0$ and $\ri\ge\lambda \rf$, we have 
\begin{align*}
    &(\lambda \rf  \ell)/\left(\frac{(\kf)^2 \lambda^2 \rf}{\kf \ri - \rf \ri + \kf \lambda \rf}\ell\right)\\&=\frac{\ri(\kf-\rf)+\lambda\kf\rf   }{\lambda\kf(\kf-\rf)+\lambda\kf\rf}< 1.
\end{align*}
       \item If $\ki > \ri$, $\lambda(\kf)^2-(\lambda-1)(\kf-\rf)\ri\ge 0$ and $\ri<\lambda \rf$, we have
\begin{align*}
    &\frac{\lambda \kf \ell-\ri\ell\left(\frac{\kf}{\rf}-1\right)}{\displaystyle\frac{(\kf)^2 \lambda^2 \rf}{\kf \ri - \rf \ri + \kf \lambda \rf}\ell}\\
    &=\frac{(\lambda\kf\rf)^2-(\ri(\kf-\rf))^2}{(\lambda\kf\rf)^2}< 1.
\end{align*}
       \item If $\ri < (\lambda-1)\rf + \kf$, $\lambda(\kf)^2 - (\lambda-1)(\kf-\rf)\ri \le 0$ and $\ri\ge\lambda \rf$, we have 
\begin{align*}
   &\ri\ell+ 
\frac{ \lambda\kf[(\lambda-1)\rf-( \ri - \kf)]\ell }{(\lambda-1)\rf+\kf}\\
&=\ri\ell+\lambda\kf\left(1-\frac{\ri}{(\lambda-1)\rf+\kf}\right)\ell>\lambda \rf  \ell.
\end{align*}
     \item If $\ri < (\lambda-1)\rf + \kf$, $\lambda(\kf)^2 - (\lambda-1)(\kf-\rf)\ri \le 0$ and $\ri<\lambda \rf$, we have $ \frac{\lambda\kf}{(\lambda-1)\rf+\kf}<\frac{\kf}{\rf}$, and then
    \begin{align*}
   &\ri\ell+ 
\frac{ \lambda\kf[(\lambda-1)\rf-( \ri - \kf)]\ell }{(\lambda-1)\rf+\kf}\\&=\lambda\kf\ell-\ri\ell\left(\frac{\lambda\kf}{(\lambda-1)\rf+\kf}-1\right)\\&>\lambda\kf\ell-\ri\ell\left(\frac{\kf}{\rf}-1\right).
 \end{align*}
 \item If  $\ki> \ri\ge (\lambda-1)\rf+\kf$ and $\lambda(\kf)^2-(\lambda-1)(\kf-\rf)\ri\le0$, we have
    \begin{align*}
   \lambda\rf\ell/\left(\frac{(\lambda-1)\rf\ri\ell}{\ri-\kf}\right)=\frac{\lambda\ri-\ki}{(\lambda-1)\ri}<1.
 \end{align*}
    \end{enumerate}
   So the values of ~\eqref{thm_4} as above are strictly smaller than our bound, equivalently, our bound is strictly tighter.
\end{itemize}

In future work, we plan to develop explicit code constructions that achieve the lower bound established in~\cref{bound_3}.  
In Appendix~\ref{ppendix_A}, we provide a concrete example where the initial code $\mathcal{C}^I$ is an $[n^I=7,k^I=4,\ell=7]$ MDS array code and the final code $\mathcal{C}^F$ is an $[n^F=3,k^F=2,\ell=7]$ MDS array code.
In this example, the conversion downloads $\frac{\lambda^2(\kf)^2  \rf}{\kf \ri - \rf \ri + \lambda\kf  \rf} \ell = 16$ subsymbols, exactly matching the lower bound in~\cref{bound_3}.

 \appendix
 \section{An Example Achieving Our Bound in Theorem 3}\label{ppendix_A}
 We provide a concrete example demonstrating that the second case of the lower bound derived in~\cref{bound_3} is achievable.
 Let $\F=\mathbb{F}_{23}$. Consider the initial MDS array code
 $$\mathcal{C}^I:[\nii,\ki,\ell]=[7,4,7],$$
 with the generator matrix
 $$ \mathbf{G}^I=\begin{bmatrix}
   \mathbf{I}_{28}&\mathbf{B} 
 \end{bmatrix},$$
 where
      $$\mathbf{B}=\begin{bmatrix}
      \mathbf{B}_{1,1}&\mathbf{B}_{1,2}&\mathbf{B}_{1,3}\\
      \mathbf{B}_{2,1}&\mathbf{B}_{2,2}&\mathbf{B}_{2,3}\\
      \mathbf{B}_{3,1}&\mathbf{B}_{3,2}&\mathbf{B}_{3,3}\\
      \mathbf{B}_{4,1}&\mathbf{B}_{4,2}&\mathbf{B}_{4,3}
  \end{bmatrix},$$
  and
  $$\mathbf{B}_{1,1}=
\begin{bmatrix}
11 & 0 & 0 & 0 & 0 & 0 & 0\\
11 & 20 & 16 & 22 & 18 & 13 & 14\\
20 & 7 & 12 & 21 & 19 & 7 & 7\\
16 & 17 & 11 & 4 & 13 & 18 & 15\\
9 & 16 & 9 & 20 & 20 & 13 & 10\\
3 & 5 & 6 & 16 & 1 & 10 & 13\\
10 & 11 & 1 & 15 & 1 & 2 & 6
\end{bmatrix},
\quad \mathbf{B}_{1,2}=
\begin{bmatrix}
3 & 0 & 0 & 0 & 0 & 0 & 0\\
10 & 0 & 17 & 17 & 19 & 13 & 20\\
13 & 19 & 1 & 12 & 19 & 0 & 13\\
20 & 6 & 21 & 17 & 2 & 11 & 10\\
1 & 21 & 0 & 17 & 6 & 7 & 0\\
15 & 12 & 13 & 4 & 3 & 2 & 19\\
20 & 21 & 4 & 22 & 13 & 8 & 18
\end{bmatrix},$$
$$
\mathbf{B}_{1,3}=
\begin{bmatrix}
7 & 0 & 0 & 0 & 0 & 0 & 0\\
9 & 6 & 7 & 22 & 11 & 4 & 7\\
22 & 21 & 10 & 13 & 2 & 11 & 7\\
0 & 16 & 19 & 1 & 8 & 5 & 14\\
8 & 3 & 6 & 22 & 11 & 16 & 3\\
15 & 0 & 6 & 10 & 15 & 18 & 4\\
18 & 5 & 9 & 0 & 16 & 20 & 1
\end{bmatrix},\quad
\mathbf{B}_{2,1}=
\begin{bmatrix}
11 & 0 & 0 & 0 & 0 & 0 & 0\\
11 & 20 & 16 & 22 & 13 & 13 & 12\\
20 & 7 & 12 & 21 & 22 & 19 & 4\\
16 & 17 & 11 & 4 & 16 & 22 & 9\\
9 & 16 & 9 & 20 & 10 & 14 & 21\\
3 & 5 & 6 & 16 & 14 & 22 & 13\\
10 & 11 & 1 & 15 & 11 & 20 & 16
\end{bmatrix},$$
$$
\mathbf{B}_{2,2}=
\begin{bmatrix}
7 & 0 & 0 & 0 & 0 & 0 & 0\\
10 & 0 & 17 & 17 & 8 & 14 & 18\\
13 & 19 & 1 & 12 & 17 & 0 & 10\\
20 & 6 & 21 & 17 & 14 & 13 & 6\\
1 & 21 & 0 & 17 & 20 & 18 & 20\\
15 & 12 & 13 & 4 & 1 & 21 & 21\\
20 & 21 & 4 & 22 & 7 & 4 & 18
\end{bmatrix},
\quad
\mathbf{B}_{2,3}=
\begin{bmatrix}
13 & 0 & 0 & 0 & 0 & 0 & 0\\
9 & 6 & 7 & 22 & 13 & 18 & 17\\
22 & 21 & 10 & 13 & 4 & 15 & 21\\
0 & 16 & 19 & 1 & 15 & 5 & 8\\
8 & 3 & 6 & 22 & 20 & 9 & 8\\
15 & 0 & 6 & 10 & 18 & 21 & 9\\
18 & 5 & 9 & 0 & 9 & 1 & 8
\end{bmatrix},$$
$$
\mathbf{B}_{3,1}=
\begin{bmatrix}
12 & 0 & 0 & 0 & 0 & 0 & 0\\
17 & 0 & 19 & 7 & 8 & 21 & 9\\
5 & 1 & 19 & 3 & 9 & 9 & 14\\
12 & 4 & 5 & 9 & 13 & 15 & 4\\
18 & 4 & 3 & 11 & 13 & 10 & 12\\
5 & 20 & 18 & 12 & 17 & 11 & 2\\
16 & 7 & 13 & 14 & 13 & 4 & 15
\end{bmatrix},
\quad
\mathbf{B}_{3,2}=
\begin{bmatrix}
14 & 0 & 0 & 0 & 0 & 0 & 0\\
6 & 5 & 7 & 11 & 9 & 13 & 22\\
11 & 3 & 6 & 21 & 11 & 14 & 4\\
7 & 0 & 1 & 16 & 20 & 14 & 19\\
11 & 19 & 5 & 5 & 21 & 13 & 5\\
0 & 13 & 2 & 11 & 19 & 10 & 1\\
4 & 16 & 7 & 11 & 13 & 15 & 11
\end{bmatrix},$$
$$
\mathbf{B}_{3,3}=
\begin{bmatrix}
9 & 0 & 0 & 0 & 0 & 0 & 0\\
18 & 1 & 16 & 22 & 18 & 10 & 10\\
0 & 0 & 4 & 18 & 20 & 22 & 13\\
20 & 7 & 16 & 14 & 15 & 18 & 13\\
18 & 17 & 15 & 19 & 10 & 9 & 5\\
11 & 14 & 3 & 10 & 10 & 19 & 13\\
2 & 17 & 8 & 9 & 8 & 1 & 19
\end{bmatrix},
\quad
\mathbf{B}_{4,1}=
\begin{bmatrix}
6 & 0 & 0 & 0 & 0 & 0 & 0\\
17 & 0 & 19 & 7 & 12 & 0 & 2\\
5 & 1 & 19 & 3 & 18 & 12 & 12\\
12 & 4 & 5 & 9 & 9 & 14 & 16\\
18 & 4 & 3 & 11 & 14 & 5 & 4\\
5 & 20 & 18 & 12 & 3 & 13 & 15\\
16 & 7 & 13 & 14 & 5 & 11 & 4
\end{bmatrix},$$
$$
\mathbf{B}_{4,2}=
\begin{bmatrix}
9 & 0 & 0 & 0 & 0 & 0 & 0\\
6 & 5 & 7 & 11 & 10 & 1 & 10\\
11 & 3 & 6 & 21 & 4 & 8 & 2\\
7 & 0 & 1 & 16 & 9 & 17 & 7\\
11 & 19 & 5 & 5 & 3 & 9 & 6\\
0 & 13 & 2 & 11 & 4 & 5 & 11\\
4 & 16 & 7 & 11 & 14 & 15 & 22
\end{bmatrix},
\quad
\mathbf{B}_{4,3}=
\begin{bmatrix}
15 & 0 & 0 & 0 & 0 & 0 & 0\\
18 & 1 & 16 & 22 & 9 & 9 & 2\\
0 & 0 & 4 & 18 & 10 & 14 & 10\\
20 & 7 & 16 & 14 & 8 & 12 & 18\\
18 & 17 & 15 & 19 & 10 & 3 & 21\\
11 & 14 & 3 & 10 & 17 & 12 & 10\\
2 & 17 & 8 & 9 & 15 & 9 & 0
\end{bmatrix}.$$
 Let the final code be the MDS array code
 $$\mathcal{C}^F:[\nf,\kf,\ell]=[3,2,7],$$
 with the generator matrix
 $$\mathbf{G}^F=\begin{bmatrix}
     \mathbf{I}_{14}&\mathbf{C}
 \end{bmatrix}, \text{where}\ 
 \mathbf{C}=\begin{bmatrix}
      \mathbf{C}_{1,1}\\
      \mathbf{C}_{2,1}
  \end{bmatrix},$$
and \[\mathbf{C}_{1,1} =
\begin{bmatrix}
16 & 21 & 8 & 21 & 20 & 4 & 20\\
8 & 3 & 2 & 4 & 6 & 1 & 22\\
20 & 1 & 1 & 7 & 6 & 2 & 15\\
1 & 0 & 17 & 16 & 21 & 18 & 17\\
17 & 2 & 0 & 17 & 21 & 20 & 15\\
7 & 2 & 18 & 2 & 1 & 8 & 4\\
2 & 20 & 22 & 15 & 7 & 1 & 3
\end{bmatrix},
\quad
\mathbf{C}_{2,1} =
\begin{bmatrix}
1 & 2 & 5 & 16 & 16 & 20 & 5\\
8 & 3 & 2 & 4 & 6 & 1 & 22\\
20 & 1 & 1 & 7 & 6 & 2 & 15\\
1 & 0 & 17 & 16 & 21 & 18 & 17\\
17 & 2 & 0 & 17 & 21 & 20 & 15\\
7 & 2 & 18 & 2 & 1 & 8 & 4\\
2 & 20 & 22 & 15 & 7 & 1 & 3
\end{bmatrix}.\]Both $\mathcal{C}^I$ and $\mathcal{C}^F$ satisfy the MDS property.
 Define the read subsymbols as
 \begin{equation*}
    D_i =
\begin{cases}
\{1\}& \text{if}\quad i\in\{1,2,3,4\},\\[2mm]
\{1,2,3,4\} & \text{if}\quad i\in\{5,6,7\}.
\end{cases}
\end{equation*}
That is, we read one  subsymbol from each systematic symbol and  four subsymbols from each of the last three symbols during conversion. Under this configuration, one can verify that (\ref{character_eq}) in Lemma~\ref{character_lemma} holds. In particular, we have 
 \begin{equation*}
     \tilde{\mathbf{C}}=\tilde{\mathbf{B}}\mathbf{E},
 \end{equation*}
 where  $\mathbf{E}$ is
 {\setcounter{MaxMatrixCols}{20}
\[
\begingroup
\setlength{\arraycolsep}{3.5pt}
\begin{bmatrix}
0 & 5 & 11 & 1 & 11 & 8 & 17 & 3 & 22 & 6 & 3 & 13 & 15 & 3\\
19 & 14 & 5 & 2 & 4 & 8 & 10 & 0 & 18 & 18 & 0 & 21 & 18 & 19\\
22 & 12 & 16 & 0 & 20 & 21 & 15 & 5 & 13 & 3 & 10 & 17 & 0 & 3\\
11 & 16 & 18 & 4 & 9 & 15 & 12 & 6 & 22 & 6 & 14 & 10 & 3 & 20\\
9 & 9 & 18 & 7 & 3 & 18 & 5 & 15 & 20 & 6 & 15 & 12 & 4 & 22\\
15 & 16 & 21 & 0 & 16 & 18 & 9 & 12 & 22 & 14 & 0 & 20 & 19 & 7\\
10 & 2 & 2 & 9 & 8 & 18 & 6 & 22 & 9 & 13 & 16 & 4 & 14 & 19\\
20 & 8 & 22 & 16 & 22 & 13 & 4 & 19 & 20 & 22 & 21 & 16 & 10 & 19\\
1 & 5 & 16 & 0 & 1 & 13 & 11 & 22 & 3 & 22 & 13 & 13 & 12 & 15\\
4 & 0 & 13 & 8 & 16 & 0 & 0 & 6 & 7 & 4 & 2 & 5 & 17 & 20\\
20 & 12 & 2 & 5 & 5 & 21 & 19 & 18 & 15 & 6 & 22 & 13 & 0 & 6\\
14 & 0 & 16 & 0 & 13 & 21 & 10 & 3 & 22 & 7 & 22 & 5 & 7 & 17
\end{bmatrix}
\endgroup
\]
}
By Remark~\ref{compute_bandwidth}, there exists a conversion procedure $\mathcal{T}$ with read bandwidth cost  $16$ subsymbols. 
\bibliographystyle{IEEEtran}
\bibliography{ref}

@inproceedings{kadekodi2019cluster,
  author    = {Kadekodi, Saurabh and Rashmi, K. V. and Ganger, Gregory R.},
  title     = {Cluster Storage Systems Gotta Have {HeART}: Improving Storage Efficiency by Exploiting Disk-Reliability Heterogeneity},
  booktitle = {17th USENIX Conference on File and Storage Technologies (FAST 19)},
  pages     = {345--358},
  year      = {2019}
}

@inproceedings{maturana2020convertible,
  author    = {Maturana, Francisco and Rashmi, K. V.},
  title     = {Convertible Codes: New Class of Codes for Efficient Conversion of Coded Data in Distributed Storage},
  booktitle = {11th Innovations in Theoretical Computer Science Conference (ITCS 2020)},
  volume    = {151},
  pages     = {66:1--66:26},
  year      = {2020}
}

@article{maturana2022convertible,
  author    = {Maturana, Francisco and Rashmi, K. V.},
  title     = {Convertible Codes: Enabling Efficient Conversion of Coded Data in Distributed Storage},
  journal   = {IEEE Transactions on Information Theory},
  volume    = {68},
  number    = {7},
  pages     = {4392--4407},
  year      = {2022}
}

@inproceedings{Maturana2022BandwidthCO,
  author    = {Maturana, Francisco and Rashmi, K. V.},
  title     = {Bandwidth Cost of Code Conversions in the Split Regime},
  booktitle = {2022 IEEE International Symposium on Information Theory (ISIT)},
  pages     = {3262--3267},
  year      = {2022}
}

@inproceedings{Maturana2020isit,
  author    = {Maturana, Francisco and Mukka, V. S. Chaitanya and Rashmi, K. V.},
  title     = {Access-Optimal Linear {MDS} Convertible Codes for All Parameters},
  booktitle = {2020 IEEE International Symposium on Information Theory (ISIT)},
  pages     = {577--582},
  year      = {2020}
}

@inproceedings{Chopra2024OnLF,
  author    = {Chopra, Saransh and Maturana, Francisco and Rashmi, K. V.},
  title     = {On Low Field Size Constructions of Access-Optimal Convertible Codes},
  booktitle = {2024 IEEE International Symposium on Information Theory (ISIT)},
  pages     = {1456--1461},
  year      = {2024}
}

@article{Kong2024,
  author    = {Kong, Xiangliang},
  title     = {Locally Repairable Convertible Codes With Optimal Access Costs},
  journal   = {IEEE Transactions on Information Theory},
  volume    = {70},
  number    = {9},
  pages     = {6239--6257},
  year      = {2024}
}

@article{Maturana2023,
  author    = {Maturana, Francisco and Rashmi, K. V.},
  title     = {Bandwidth Cost of Code Conversions in Distributed Storage: Fundamental Limits and Optimal Constructions},
  journal   = {IEEE Transactions on Information Theory},
  volume    = {69},
  number    = {8},
  pages     = {4993--5008},
  year      = {2023}
}

@inproceedings{Maturana2023isit,
  author    = {Maturana, Francisco and Rashmi, K. V.},
  title     = {Locally Repairable Convertible Codes: Erasure Codes for Efficient Repair and Conversion},
  booktitle = {2023 IEEE International Symposium on Information Theory (ISIT)},
  pages     = {2033--2038},
  year      = {2023}
}

@article{Ge2024MDSGC,
  author        = {Ge, Songping and Cai, Han and Tang, Xiaohu},
  title         = {{MDS} Generalized Convertible Code},
  journal = {arXiv preprint arXiv:2407.14304},
  year          = {2024}
}

@article{Shi2025BoundsAO,
  author  = {Shi, Haoming and Fang, Weijun and Gao, Yuan},
  title   = {Bounds and Optimal Constructions of Generalized Merge-Convertible Codes for Code Conversion Into {LRC}s},
  journal = {IEEE Transactions on Information Theory},
  year    = {2026},
  volume={72},
  number={5},
  pages={2841--2860}
}

@ARTICLE{Ge2025LocallyRC,
  author={Ge, Songping and Cai, Han and Tang, Xiaohu},
  journal={IEEE Transactions on Information Theory}, 
  title={Locally Repairable Convertible Codes: Improved Lower Bound and General Construction}, 
  year={2026},
  volume={72},
  number={5},
  pages={2915--2930}
  }

@incollection{weatherspoon02,
  author    = {Weatherspoon, Hakim and Kubiatowicz, John D.},
  title     = {Erasure Coding Vs. Replication: A Quantitative Comparison},
  booktitle = {Peer-to-Peer Systems},
  volume    = {2429},
  pages     = {328--337},
  year      = {2002}
}

@ARTICLE{RKSVK25,
  author={Ramkumar, Vinayak and Kong, Xiangliang and Yeswanth Sai, G. and Vajha, Myna and Nikhil Krishnan, M.},
  journal={IEEE Transactions on Information Theory}, 
  title={On {MDS} Convertible Codes in the Merge Regime}, 
  year={2026},
  volume={72},
  number={5},
  pages={2963--2977}
}

@article{singhvi2025tight,
  author        = {Singhvi, Shubhransh and Chopra, Saransh and Rashmi, K. V.},
  title         = {Tight Lower Bounds on the Bandwidth Cost of {MDS} Convertible Codes in the Split Regime},
 journal = {arXiv preprint arXiv:2511.12279},
  year          = {2025}
}

@article{Gruica2026ConvertibleCF,
  author        = {Gruica, Anina and Jany, Benjamin and Kruglik, Stanislav},
  title         = {Convertible Codes for Data and Device Heterogeneity},
  journal = {arXiv preprint arXiv:2601.10341},
  year          = {2026}
}

@inproceedings{Zhang25,
  author    = {Zhang, Justin and Rashmi, K. V.},
  title     = {Secure Convertible Codes for Passive Eavesdroppers},
  booktitle = {2025 IEEE International Symposium on Information Theory (ISIT)},
  pages     = {2126--2131},
  year      = {2025}
}
\end{document}